\documentclass{amsart}
\usepackage{url}
\usepackage{xcolor}

\newcommand{\Z}{\mathbb Z}

\newcommand{\N}{\mathbb N}
\newcommand{\F}{\mathbb F}

\newcommand{\XOR}{\mathbin{\mathtt{XOR}}}
\newcommand{\OR}{\mathbin{\mathtt{OR}}}
\newcommand{\AND}{\mathbin{\mathtt{AND}}}
\newcommand{\EXP}{\mathbin{\uparrow}}
\newcommand{\DIV}{\mathbin{/\!/}}
\newcommand{\NOT}{\operatorname{\mathtt{NOT}}}
\newcommand{\ord}{\operatorname{ord}}
\newcommand{\mask}{\operatorname{\mathtt{MASK}}}

\renewcommand{\:}{\colon}
\renewcommand{\>}{\rightarrow}

\newtheorem{theorem}{Theorem}
\theoremstyle{plain}

\newtheorem{corollary}{Corollary}

\newtheorem{example}{Example}

\newtheorem{lemma}{Lemma}

\newtheorem{proposition}{Proposition}
\newtheorem{remark}{Remark}

\begin{document}
\title[T-functions revisited]{T-functions revisited: New criteria for bijectivity/transitivity}
\author{Vladimir Anashin}
\author{Andrei Khrennikov} 
\author{Ekaterina Yurova}
\thanks{The authors were supported in part by  grant of
 Linnaeus University Mathematical Modelling. The first of the authors was
 also supported by 
 Russian Foundation for Basic Research grant No 12-01-00680-a
 and by Chinese Academy of Sciences visiting professorship for senior international
scientists grant No  2009G2-11}


\begin{abstract}
The paper presents new criteria for bijectivity/transitivity of
T-functions and a fast knapsack-like algorithm of evaluation of a T-function.
Our approach is based on non-Archimedean ergodic theory: Both the criteria and algorithm use van der Put series to represent 1-Lipschitz $p$-adic functions
and to study measure-preservation/ergodicity of these.
\end{abstract}

\maketitle

\section{Introduction}

For years  linear feedback shift
registers (LFSRs) over a 2-element field $\mathbb F_2$ have been one of the most important building blocks in
keystream generators of stream ciphers. LFSRs can easily be
designed to produce binary sequences of the longest period (that is, of length $2^k-1$
for a $k$-cell LFSR over $\mathbb F_2$); LFSRs are fast and easy to
implement both in hardware and in software. However, sequences produced by LFSRs  have  linear dependencies that make easy to analyse  the sequences to construct
attacks on the whole cipher. To make output sequences of LFSRs more secure
these linear dependencies must be destroyed by a properly chosen
filter; this is the filter that carries the major cryptographical load making
the whole cipher secure. 

T-functions were found to be useful tools to design fast
cryptographic primitives and ciphers based on usage of both arithmetic (addition, multiplication) and logical operations, see \cite{ASCCipher,hong05new,Hong05tsc3,KlSh,klimov03cryptographic,klimov04new,klimov05app,klimov05thesis,klsh04tfi,kotomina99,Mir1,ABCv3,tsc4Cipher,vest06phase2Cipher,vestwebsite}.
Loosely speaking, a T-function is a map of $k$-bit words into $k$-bit words such
that each $i$-th bit of image depends only on low-order bits $0,...,i$ of the
pre-image. Various methods are known to construct bijective T-functions as
well as transitive T-functions
(the latter are the ones that produce sequences of the longest possible period, $2^k$),
see \cite{AnKhr,me:1,me:2,me:conf,me:ex,me-CJ,me-NATO,kotomina99,Lar,KlSh,klimov03cryptographic,klimov04new,klimov05thesis,hong05new}. Transitive T-functions have been considered as a candidate
to replace LFSRs in keystream generators of stream ciphers, see e.g. \cite{ABCv3,ASCCipher,Hong05tsc3,klsh04tfi,Mir1,tsc4Cipher} since sequences produced by T-function-based keystream generators 
are proved to have a number
of good cryptographic properties, e.g.,  high linear and 2-adic complexity,
uniform distribution of subwords, etc., see \cite{AnKhr,Koloko,me-NATO,LinCompT-fun}.
Bijective T-functions can be  used to design  filter functions
of stream ciphers, see \cite{AnKhr,me-NATO,ABCv3}. 

The main purpose of this paper is to provide new criteria for bijectivity/transitivity
of T-functions; so Theorems \ref{thm:vdp-mespres} and \ref{thm:ergnew}
are main results of the paper. In our opinion, these new criteria  might
be better applicable to T-functions
that are represented via compositions of standard computer instructions 
on representation of T-functions via additions of some non-negative constants.
What is important, the representation can be used to evaluate a T-function
via a knapsack-like algorithm:
Assuming the constants are stored
in memory,   \emph{to calculate a value of an arbitrary T-function on a $k$-bit word, one needs only not more than $k$ calls to memory and not more than $k-1$ additions modulo $2^k$ of $k$-bit numbers}. Thus, representations of that sort could be used in design of various high-performance cryptographic primitives:
keystream generators, filter functions, cipher combiners (Latin squares) and fast stream
ciphers as a whole. 

The representation is based on \emph{van der Put series}, special convergent series
from $p$-adic analysis; that is why
the $p$-adic analysis (and $p$-adic ergodic theory) are main  mathematical tools we use in the paper:
To determine bijectivity/transitivity of a given T-function, we represent
it via van der Put series \eqref{vdp} and apply accordingly Theorems
\ref{thm:vdp-mespres} and/or \ref{thm:ergnew}. We stress that given a T-function combined from basic computer instructions, `normally'
it is easier 
to represent it via van der Put series than via Mahler series (or via
coordinate functions $\psi_i$ to use criteria based on algebraic normal forms,
see Subsection \ref{ssec:erg}); moreover, once the T-function  is represented via van der Put series, it is much faster to evaluate it compared to representation
via Mahler series or via coordinate functions.

The van der Put series may also play an important role in a study
of linear dependencies among coordinate sequences (the ones produced by the
mentioned coordinate functions $\psi_i$) of a transitive T-function. Although  given
a `randomly chosen'  T-function, its coordinate sequences produced by $\psi_i$
and $\psi_j$ should be considered as `independent' once $i\ne j$ (meaning
the first half-periods of these sequences are independent Boolean vectors, cf. \cite[Theorem 11.26]{AnKhr}), this is not the case for large classes
of transitive T-functions: There are linear relations between any two
adjacent coordinate sequences in Klimov-Shamir T-functions \cite{LinearKlSh05,LinPropT-func},
in  polynomials with integer coefficients \cite{LinPropPol}, and in T-functions that are uniformly differentiable modulo 4, \cite{me:Linrel_tf-SETA2012,me:Linrel_tf-prep}.
The latter class is currently the largest known class where transitive T-functions
exhibit linear relations between adjacent coordinate sequences.
Therefore an important problem is to characterize transitive T-functions
that exhibit no linear dependencies between adjacent coordinate sequences
as the said linear relations may result in attacks against T-function-based
ciphers, \cite{me:Linrel_tf-prep}.
Van der Put series might be an adequate tool in a study of the problem as
with the use of the series one can handle `non-smooth' T-functions  (whence, the ones which are not
uniformly differentiable modulo 4). We note however that the mentioned study is a future work whose subject is outside
the scope of the current paper.

It is worth noticing here that the $p$-adic ergodic theory which is 
exploited in the  current paper constitutes an important  part
of non-Archimedean (and wider, of algebraic)
dynamics, a rapidly developing mathematical discipline that recently demonstrated
its effectiveness in application to various sciences: computer science,
cryptology,
physics, molecular biology, cognitive sciences, genetics, artificial
intelligence, image processing, numerical analysis and modelling, etc. Due
to a huge number of papers in the area, we can 
mention only some monographs  here to enable the interested reader to find
relevant references therein: \cite{AnKhr,Khrennikov9,Khrennikov:1997,Khren:mono}. 

The paper is organized as follows:
\begin{itemize}
\item In Section \ref{sec:T-non_A} we give a brief survey of non-Archimedean
theory of T-functions: history, state of the art, main notions and results. 
\item In Section \ref{sec:main} we prove the said criteria for bijectivity/transitivity
of T-functions in terms of van der Put series.
\item Using the transitivity criterion, in Section \ref{sec:App} we give two examples
of ergodic T-functions which are composed of additions and maskings. Also
we
explain how to use the bijectivity criterion
in order to construct huge classes of large Latin squares and
introduce a fast knapsack-like  algorithm of evaluation of a T-function
represented by van der Put series.
\item We conclude in Section \ref{sec:Concl}.
\end{itemize}
 
\section{Non-Archimedean theory of T-functions: brief survey}
\label{sec:T-non_A}
In this Section, we introduce basics of what can be called the non-Archimedean
approach to T-functions. For the full theory see monograph \cite{AnKhr} or
 expository paper \cite{me-NATO}.
 We start with a definition of a T-function and show that
T-functions can be treated as continuous
functions defined on and valued in the space of 2-adic integers. Therefore
we introduce basics of 2-adic arithmetic and of 2-adic Calculus that we will
need to state and prove our main results. There are many comprehensive monographs
on  $p$-adic numbers and $p$-adic analysis that contain all necessary definitions
and proofs, see e.g. \cite{Kobl,Mah,Sch}
or introductory chapters in \cite{AnKhr}; so further in the Section we introduce
2-adic numbers in a somewhat informal manner.


\subsection{T-functions}

An $n$-variate triangular function (a T-function for
short) is a mapping
\[
\left(\alpha_0^{\downarrow},\alpha_1^{\downarrow},
\alpha_2^{\downarrow},\ldots\right)\mapsto
\left(\Phi_0^{\downarrow}\left(\alpha_0^{\downarrow}\right),
\Phi_1^{\downarrow}\left(\alpha_0^{\downarrow},
\alpha_1^{\downarrow}\right),
\Phi_2^{\downarrow}\left(\alpha_0^{\downarrow},
\alpha_1^{\downarrow},\alpha_2^{\downarrow}\right),
\ldots\right),
\]
where $\alpha_i^{\downarrow}\in\F_2^n$ is a Boolean columnar
$n$-dimensional vector; $\F_2=\{0,1\}$ is a 2-element field, and
\[
\Phi_i^{\downarrow}\colon (\F_2^n)^{i+1}\to \F_2^m
\]
maps $(i+1)$ Boolean columnar $n$-dimensional vectors
$\alpha_0^{\downarrow},\ldots,\alpha_i^{\downarrow}$ to
$m$-dimensional columnar Boolean vector
$\Phi_i^{\downarrow}\left(\alpha_0^{\downarrow},
\ldots,\alpha_i^{\downarrow}\right)$. Accordingly, a
{univariate T-function $f$} is a mapping
\begin{equation}
\label{eq:T-uni}
(\chi_0,\chi_1,\chi_2, \ldots)\stackrel{f}{\mapsto}
(\psi_0(\chi_0);\psi_1(\chi_0,\chi_1);
\psi_2(\chi_0,\chi_1,\chi_2);\ldots),
\end{equation}
where $\chi_j\in\{0,1\}$, and each
$\psi_j(\chi_0,\ldots,\chi_j)$ is a Boolean function in
Boolean variables $\chi_0,\ldots,\chi_j$. T-functions may
be viewed as mappings from non-negative integers to
non-negative integers: e.g., a univariate T-function $f$
sends a number with the base-$2$ expansion
\[
\chi_0+\chi_1\cdot 2+\chi_2\cdot 2^2+\cdots
\]
to the number with the base-2 expansion
\[
\psi_0(\chi_0)+\psi_1(\chi_0,\chi_1)\cdot
2+\psi_2(\chi_0,\chi_1,\chi_2)\cdot 2^2+\cdots
\]
Further in the paper we refer to these Boolean functions
$\psi_0, \psi_1,\psi_2,\ldots$ as \emph{coordinate
functions} of a T-function $f$. If we restrict
T-functions to the set of all numbers whose base-$2$
expansions are not longer than $k$, we sometimes refer these
restrictions as \emph{T-functions on $k$-bit words}. Thus, we may consider
the restriction of a (univariate)
T-function $f$  on $k$-bit words as a transformation $f\bmod 2^k$ of the residue ring $\Z/2^k\Z$
modulo $2^k$ as every residue $r\in\Z/2^k\Z$ can be associated to a base-2
expansion of a non-negative
integer from $\{0,1,\ldots,2^k-1\}$. 

Important examples of T-functions are basic machine
instructions:
\begin{itemize}
\item integer arithmetic operations (addition,
multiplication,\ldots);
\item bitwise logical operations ($\OR$, $\XOR$, $\AND$,
$\NOT$);
\item some of their compositions (masking with mask $m$, $\mask(\cdot,m)=\cdot\AND m$; $\ell$-bit
shifts towards higher
order bits, $\cdot 2^\ell$; reduction modulo $2^k$, $\cdot\bmod 2^k=\cdot\AND(2^k-1)$).
\end{itemize}
Since obviously a composition of T-functions is a
T-function (for instance, \emph{any polynomial with
integer coefficients is a T-function}), the T-functions
are natural transformations of binary words that can be performed by digital computers. 
That is the main reason why T-functions  have attracted attention of
cryptographic community: T-functions may be used to construct new cryptographic
primitives
that are suitable for software implementations. For this purpose, given a
T-function $f$ and a bitlength  $k$, it is important to determine whether the mapping $f\bmod 2^k$ is bijective (that is, invertible) or transitive (that is, a permutation with a single cycle).

Although in cryptology the term ``T-function'' was
suggested only in 2002, by Klimov and Shamir,
see~\cite{KlSh}, in mathematics 
these mappings were known (however, under other names) long before 2002, and various effective criteria were proved to determine bijectivity and/or transitivity of T-functions. Some of these criteria were published in 1994
by the first author of the paper, see \cite{me:1}; the criteria are based on
2-adic ergodic theory and use representation of a T-function $f$ in the
form of \emph{Mahler series} \eqref{eq:Mah}.
 However, in some
cases to represent given T-function in this form may be  a difficult mathematical
task by itself; that is why a variety of criteria is needed to handle T-functions
represented as compositions of various computer instructions (we give a brief
survey of known criteria in Subsection \ref{ssec:erg}).
\subsection{Brief history of T-functions}
As said, in cryptology the term ``T-function'' was
suggested only in 2002, 
in mathematics the mappings we
refer to as T-functions have being studied more than half a century.

In automata theory, Yablonsky et. al. studied
the so-called \emph{determinate functions} since 1950-th,
see~\cite{Yb}. Actually, a determinate function is a
function that can be represented by an automaton: Consider
an automaton with a binary input and binary output; then the automaton
transforms
each infinite input string of $0$-s and $1$-s  into  infinite output string of $0$-s and $1$-s
(we suppose that the initial state of the automaton is fixed). Note that
every outputted $i$-th bit depends only on the inputted
$i$-th bit and on the current state of the automaton. However,
the current state depends only on the previous state and on
the $(i-1)$-th input bit. Hence, for every $i=1,2,\ldots$,
the $i$-th outputted bit depends only on bits $1,2,\ldots,i$
of the input string, and so the transformation of all
infinite binary strings performed by the automaton is a
T-function.  The $p$-adic approach in automata theory goes back to the work
of Lunts \cite{Lunts} of 1965.

We note here that Yablonsky and his succeeders were mostly
interested in such properties of determinate functions as
completeness of various systems of functions, various methods how to
construct an automaton that represents the given function, etc. It is
worth noticing also that determinate functions were studied in
a more general setting, for arbitrary $K$-letter
inputs/outputs, cf.~\cite{Yb} or recent works \cite{me-auto_fin,me:Discr_Syst} and references
therein. It is worth noticing here that in \cite{me-auto_fin} (as well as
in the current paper) van der Put series is the main tool.

T-functions were also studied in algebra, and also in much
more general setting: T-functions are a special case of the
so-called \emph{compatible mappings}; namely, any
T-function on $k$-bit words is a compatible mapping of a
residue ring modulo~$2^k$.

Recall that a transformation $f\colon A\to A$ of an
algebraic system $A$ (that is, of a set endowed with
operations $\omega\in\Omega$) is called \emph{compatible}
whenever $f$ agrees with all congruences of $A$; that is,
$f(a)\sim f(b)$ whenever $a\sim b$, where $\sim$ is a
congruence of $A$. Recall that a congruence $\sim$ is an
equivalence relation that agrees with all operations
$\omega$ of $A$:
$\omega(a_1,\ldots,a_r)\sim\omega(b_1,\ldots,b_r)$ whenever
$a_i\sim b_i$, $i=1,2,\ldots,r$, $r$ is the arity of
$\omega$. It is obvious that a compatible transformation of
the residue ring $\Z/2^k\Z$ modulo $2^k$ is a T-function
on $k$-bit words (under a natural one-to-one correspondence
between $\F_2^k$ and $\Z/2^k\Z$, when an $k$-bit word from
$\F_2^k$ is considered as a base-$2$ expansion of the corresponding integer):
In the case when $A=\Z/2^k\Z$ is a residue ring modulo $2^k$ the
compatibility of $f$ yields
$$
a\equiv b\pmod {2^s}\implies f(a)\equiv f(b) \pmod{2^s},
$$
for all $s\le k$, which is merely an equivalent definition
of a (univariate) T-function on $k$-bit words. For further
results on compatible functions on rings and other algebras
see  monograph~\cite{LN} and references
therein. 

Since the early 1990-th the non-Archimedean theory of
T-functions emerged, which treated T-functions as
continuous transformations of the space $\Z_2$ of $2$-adic
integers and studied corresponding dynamics. The first
publications in this area were~\cite{me:1}
and~\cite{me:conf}; the importance of T-functions
for pseudorandom generation and cryptology was explicitly
stated in these papers as well. Within that theory, it was
demonstrated that T-functions are continuous
transformations with respect to  $2$-adic metric, and
that bijectivity (resp., transitivity) of
T-functions correspond to measure preservation (resp.,
ergodicity) of these continuous transformations. This
approach supplies a researcher with a number of effective
tools from the non-Archimedean (actually, $2$-adic) analysis
to determine whether a given T-function is bijective or
transitive, to study distribution and
structure of output sequences, to construct wide classes of
T-functions with prescribed properties, etc. We use
this approach in the present paper.

We note that the theory was developed in a more
general setting, for arbitrary prime $p$, and not
necessarily for $p=2$, which corresponds to the case of
T-functions. In the paper, we are mostly interested in the case
$p=2$; however, we prove corresponding results for arbitrary prime $p$ where
possible.   

Last, but not least: T-functions under the name of
triangular Boolean mappings were studied in the theory of
Boolean functions. Within this theory there were obtained
important criteria for invertibility/transitivity
of a T-function, in terms of
coordinate Boolean functions $\psi_0,\psi_1,\psi_2,\ldots$.
The criteria belongs to mathematical folklore circulating at
least since 1970-th among mathematicians dealing with the
theory of Boolean functions; unfortunately the author of these important
criteria is not known.  To the best of our knowledge, the first
quotation of these folklore criteria in literature occurred in 1994, see
\cite[Lemma 4.8]{me:1}: As it is clearly marked in the mentioned paper, the
said Lemma is just a re-statement of these well-known criteria. We reproduce these folklore criteria
below, see Theorem \ref{thm:ergBool}.

\subsection{T-functions, $2$-adic integers, and $p$-adic analysis}
\label{ssec:NA}

As it follows directly from the definition, any T-function
is well-defined on the set $\Z_2$ of all infinite binary
sequences $\ldots\delta_2(x)\delta_1(x)\delta_0(x)=x$, where
$\delta_j(x)\in\F_2=\{0,1\}$, $j=0,1,2,\ldots$. Arithmetic
operations (addition and multiplication) with these
sequences can be defined via standard ``school-textbook''
algorithms of addition and multiplication of natural numbers
represented by base-$2$ expansions. Each term of a sequence
that corresponds to the sum (respectively, to the product)
of two given sequences can be calculated by these
algorithms within a finite number of steps.

Thus, $\Z_2$ is a commutative ring with respect to the
addition and multiplication defined in this manner. The ring $\Z_2$ is
called the ring of \emph{$2$-adic integers}. The ring $\Z_2$
contains a subring $\Z$ of all rational integers: For
instance, $\ldots111=-1$.

Moreover, the ring $\Z_2$ contains all rational numbers that
can be represented by irreducible fractions with odd
denominators. For instance, 
$\ldots01010101\times \ldots00011=\ldots111$, i.e.,
 $\ldots01010101=-1/3$ since $\ldots00011=3$ and
$\ldots111=-1$.

Sequences with only a finite number of $1$-s correspond to
non-negative rational integers in their base-$2$ expansions,
sequences with only a finite number of $0$-s correspond to
negative rational integers, while eventually periodic
sequences (that is, sequences that become periodic 
from a certain place) correspond to rational numbers
represented by irreducible fractions with odd denominators:
For instance, $3=\ldots00011$, $-3=\ldots11101$,
$1/3=\ldots10101011$, $-1/3=\ldots1010101$. So the $j$-th
term $\delta_j(u)$ of the corresponding sequence $u\in\Z_2$
is merely the $j$-th digit of the base-$2$ expansion of $u$
whenever $u$ is a non-negative rational integer,
$u\in\N_0=\{0,1,2,\ldots\}$.

What is important, the ring $\Z_2$ is a metric space with
respect to the metrics (distance) $d_2(u,v)$ defined by the
following rule: $d_2(u,v)=|u-v|_2=1/2^n$, where $n$ is the
smallest non-negative rational integer such that
$\delta_n(u)\ne\delta_n(v)$, and $d_2(u,v)=0$ if no such $n$
exists (i.e., if $u=v$). For instance $d_2(3,1/3)=1/8$. The
function $d_2(u,0)=|u|_2$ is a 2-adic absolute value of the $2$-adic integer
$u$, and $\ord_2 u=-\log_2|u|_2$ is a $2$-adic valuation
of $u$. Note that for $u\in\Z$ the valuation $\ord_2 u$ is
merely the exponent of the highest power of $2$ that divides
$u$ (thus, loosely speaking, $\ord_2 0=\infty$; so
$|0|_2=0$). This means, in particular, that
\begin{equation}
\label{eq:cong-dist}
|a-b|_2\le2^{-k}\quad \text{if and only if}\quad a\equiv b\pmod{2^k}
\end{equation}
for $a,b\in\Z$. Using this equivalence, one can expand the map $\bmod{2^k}$
(the reduction
modulo $2^k$) to  the whole space $\Z_2$, obtaining a T-function  
$\cdot\bmod{2^k}=\cdot\AND{(2^k-1})$ that is defined everywhere on $\Z_2$;
so further we use both (equivalent) notations $|a-b|_2\le2^{-k}$ and  $a\equiv b\pmod{2^k}$
for arbitrary $a,b\in\Z_2$.

It is easy to see that the metric $d_2$ satisfies the \emph{strong triangle
inequality}:
\begin{equation}
\label{eq:s-tri}
|a+b|_2\le\max\{|a|_2,|b|_2\},
\end{equation}
for all $a,b\in\Z_2$. Metric spaces  of this kind are called \emph{ultrametric}
spaces,
or \emph{non-Archimedean} metrics spaces; the latter due to the fact that in ultrametric
spaces the Archimedean Axiom does not hold.

Once the metric is defined, one defines notions of
convergent sequences, limits, continuous functions on the
metric space, and derivatives if the space is a commutative
ring. For instance, with respect to the so defined metric
$d_2$ on $\Z_2$ the  sequence $1,3,7,\ldots,2^i-1,\ldots$ (or, in base-2
expansion, the sequence $1,11,111, \dots$) tends to $-1=\ldots
111$.
Bitwise logical operators (such as $\XOR$, $\AND$,~\ldots)
define continuous functions in two variables. 

Reduction modulo $2^n$ of a $2$-adic integer $v$, i.e.,
setting all terms of the corresponding sequence with indices
greater than $n-1$ to zero (that is, taking the first $n$
digits in the representation of $v$) is just an
approximation of a $2$-adic integer $v$ by a rational
integer with precision $1/2^n$: This approximation is an
$n$-digit positive rational integer $v \AND (2^n-1)$; the
latter will be denoted also as $v\bmod{2^n}$.

Actually \emph{a processor operates with approximations of
$2$-adic integers with respect to $2$-adic metrics}: When the
overflow happens, i.e., when a number that must be written
into an $n$-bit register consists of more than $n$
significant bits, the processor just writes only $n$ low
order bits of the number into the register thus reducing the
number modulo $2^n$. Thus, the accuracy of the approximation is
defined by the bitlength of machine words  of the processor.

What is the most important within the scope of the paper is that
all T-functions are \emph{continuous} functions of
$2$-adic variables, since \emph{all T-functions satisfy
{Lipschitz condition with a constant $1$} with respect to
the $2$-adic metrics}, and vice versa.

Indeed, it is obvious that the function $f\colon
\Z_2\to\Z_2$ satisfies the condition
$|f(u)-f(v)|_2\le|u-v|_2$ for all $u,v\in\Z_2$ if and
only if $f$ is compatible, since the inequality
$|a-b|_2\le1/2^k$ is just equivalent to the congruence
$a\equiv b \pmod{2^k}$. A similar property holds for
multivariate T-functions.
So we conclude:
\begin{center}\em
T-functions${}={}$compatible functions${}=1$-Lipschitz
functions
\end{center}

The observation we just have made implies that the
non-Archimedean (namely, the $2$-adic) analysis can be used
in the study of T-functions. For instance, one can prove
that the following functions satisfy Lipschitz condition
with a constant $1$ and thus are T-functions (so we can
use them in compositions to construct PRNGs):
\begin{itemize}
\item subtraction, $-$: $(u,v)\mapsto u-v$; 
\item exponentiation, $\EXP$: $(u,v)\mapsto u\EXP
v=(1+2u)^v$, and in particular
raising to negative powers, $u\EXP(-v)=(1+2u)^{-v}$;
\item division, $\DIV$: $u\DIV v=u\cdot
(v\EXP(-1))=u/(1+2v)$.
\end{itemize}

We summarize:
\begin{itemize}
\item T-functions on $n$-bit words are approximations
of $2$-adic compatible functions (i.e.,
$1$-Lipschitz functions) with
precision $2^{-n}$ w.r.t. the $2$-adic metric: That is, a
T-function on $n$-bit words is just a reduction modulo
$2^n$ of a $2$-adic $1$-Lipschitz function.
\item To study properties of T-functions one can use
$2$-adic analysis , since compatible functions are
continuous w.r.t. the $2$-adic metric.
\item In addition to the basic machine instructions, to
construct T-functions one can use also subtraction,
division by an odd integer, raising an odd integer to a certain power.
\end{itemize}

All these considerations after proper modifications remain
true for arbitrary prime $p$, and \emph{not} only for $p=2$,
thus resulting the notion of the $p$-adic integer and in respective
$p$-adic analysis. For formal introduction to $p$-adic
analysis, exact notions and results see any relevant book,
e.g.~\cite{Kobl,Mah}. 
\subsection{The $2$-adic ergodic theory and bijectivity/transitivity of T-functions}
\label{ssec:erg}

Now we describe the connections between bijectivity/transitivity of T-functions and the $2$-adic
ergodic theory. We first recall some basic notions of dynamics and of ergodic
theory (which is a part of dynamics). 

A \emph{dynamical system on a
measure space} $\mathbb S$ is a triple $(\mathbb S;\mu;
f)$, where $\mathbb S$ is a set endowed with a measure
$\mu$, and $f\colon \mathbb S\to \mathbb S$ is a
\emph{measurable function}; that is, an $f$-preimage of any
measurable subset is a measurable subset. These basic
definitions from dynamical system theory, as well as the
following ones, can be found in~\cite{KN}; see
also~\cite{Kat_Has} as a comprehensive monograph on
various aspects of dynamical systems theory.

A \emph{trajectory} (or, \emph{orbit}) 
is a sequence
\[
x_0, x_1=f(x_0),\ldots, x_i=f(x_{i-1})=f^i(x_0),\ldots
\]
of points of the space $\mathbb S$, $x_0$ is called an
\emph{initial} point of the trajectory. If $F\colon \mathbb
S\to \mathbb T$ is a measurable mapping to some other
measurable space $\mathbb T$ with a measure $\nu$ (that is,
if an $F$-preimage of any $\nu$-measurable subset of
$\mathbb T$ is a $\mu$-measurable subset of $\mathbb S$),
the sequence $F(x_0), F(x_1), F(x_2),\ldots$ is called an
\emph{observable}. Note that the trajectory formally looks
like the sequence of states of a pseudorandom generator
while the observable resembles the output sequence.

A mapping $F\colon\mathbb S\to\mathbb Y$ of a measure
space $\mathbb S$ into a measure space $\mathbb Y$
endowed with probability measures $\mu$ and $\nu$,
respectively, is said to be {\emph{measure preserving}} (or,
sometimes, \emph{equiprobable}) whenever
$\mu(F^{-1}(S))=\nu(S)$ for each measurable subset
$S\subset\mathbb Y$. In the case $\mathbb S=\mathbb Y$ and
$\mu=\nu$, a measure preserving mapping $F$ is said to be
{\emph{ergodic}} whenever for each measurable subset $S$
such that $F^{-1}(S)=S$ one has either $\mu(S)=1$ or
$\mu(S)=0$.

Recall that to define a measure $\mu$ on some set $\mathbb
S$ we should assign non-negative real numbers to some
subsets that are called elementary. All other
\emph{measurable} subsets are compositions of these
elementary subsets with respect to countable unions,
intersections, and complements.

Elementary measurable subsets in $\Z_2$ are balls
$\mathbf B_{2^{-k}}(a)=a+2^k\Z_2$ of radii $2^{-k}$ centered at $a\in\Z_2$
(that is,
co-sets with respect to the ideal $2^k\Z_2$ of the ring $\Z_2$, generated by $2^k$). 
To each
ball we assign a number $\mu_2(\mathbf B_{2^{-k}}(a))=1/2^k$. This
way we define the probability measure $\mu_2$ on the space $\Z_2$,
$\mu_2(\Z_2)=1$. The measure $\mu_2$ is  a
(normalized) \emph{Haar measure} on $\Z_2$. The normalized
Haar measure on $\Z_2^n$ can be defined in a similar manner.

To put it in other words, a ball $a+2^k\Z_2$ (of radius $2^{-k}$) is
just a set of all $2$-adic integers that are congruent to $a$
modulo $2^k$; that is, the set of all infinite binary words that have 
common initial prefix of length $k$ which coincides with the one of the (infinite)
binary word $a$. The measure of this set is
$\mu_2(a+2^k\Z_2)=2^{-k}$. For example, $\cdots
*****0101=5+16\cdot\Z_2=-1/3+16\cdot\Z_2$ is a ball of
radius (and of measure) $1/16$ centered at the point $5$
(or, which is the same, at the point $-1/3$); all $2$-adic numbers
that are congruent to $5$ modulo~$16$ comprise this ball.

Note that the sequence $(s_i)_{i=0}^\infty$ of $2$-adic
integers is uniformly distributed (with respect to the
 measure $\mu_2$ on $\Z_2$) if and only if it
is uniformly distributed modulo $2^k$ for all
$k=1,2,\ldots$; That is, for every $a\in\Z/2^k\Z$ relative
numbers of occurrences of $a$ within initial segment of
length $\ell$ of the sequence $(s_i\bmod 2^k)_{i=0}^\infty$ of residues
modulo $2^k$ are asymptotically equal, i.e.,
$\lim_{\ell\to\infty}A(a,\ell)/\ell=1/2^k$, where
$A(a,\ell)=\#\{s_i\equiv a\pmod{2^k}\colon i<\ell\}$,
see~\cite{KN} for details. Thus, strictly uniformly
distributed sequences are uniformly distributed in the
usual meaning of the theory of distributions of sequences.
Of course, considerations of the above sort take place
for arbitrary prime $p$, and not only in the case when $p=2$.

The following Theorem (which was announced
in~\cite{me:2} and proved
in~\cite{me-spher}) holds: 

\begin{theorem}
\label{thm:erg-tran}
For $m=n=1$, a $1$-Lipschitz mapping
$F\colon\Z_p^n\to\Z_p^m$ {preserves the normalized Haar
measure} $\mu_p$ on $\Z_p$ (resp., is ergodic with respect
to $\mu_p$) if and only if it is bijective (resp.,
transitive) modulo $p^k$ for all $k=1,2,3,\ldots$.

For $n\ge m$, the mapping $F$ preserves the measure $\mu_p$ if
and only if it induces a balanced mapping of $(\Z/p^k\Z)^n$
onto $(\Z/p^k\Z)^m$, for all $k=1,2,3,\ldots$.
\end{theorem}

In other words, Theorem \ref{thm:erg-tran} yields that
\begin{itemize}\em
\item for a univariate T-function $f$, {measure preservation is
equivalent to bijectivity} of  $f\bmod 2^k$ for all
$k\in\N$;
\item for a multivariate T-function 
$F\colon\Z_2^n\to\Z_2^m$, $m\le n$,  measure preservation is
equivalent to a balance of $F\bmod 2^k$ for all
$k\in\N$;
\item ergodicity of $F\colon\Z_2^n\to\Z_2^n$ is equivalent to transitivity of $F\bmod 2^k$ for all $k\in\N$.
\end{itemize}
This theorem implies in particular that whenever one chooses
an ergodic T-function $f\colon\Z_2\to\Z_2$ as a state
transition function of an automaton 
and a measure-preserving
T-function $F\colon(\Z/2^k\Z)^n\to (\Z/2^k\Z)^m$ as an
output function, both the sequence of states and
output sequence of the automaton are uniformly distributed with
respect to the Haar measure. This implies that reduction of
these sequences modulo $2^n$ results in strictly uniformly
distributed sequences of binary words. Note also that any number that is
longer than a word bitlength of a computer, is
reduced modulo $2^n$ automatically.

Thus, Theorem~\ref{thm:erg-tran} points out a way to
construct generators of uniformly distributed sequences from standard computer instructions. To construct such a generator, one
must answer the following questions: {What compositions of
basic machine instructions} are measure-preserving? are
ergodic? Given a composition of basic machine instructions,
is it measure-preserving? is it ergodic?
These questions can be answered with the use of  $p$-adic ergodic theory,
see papers \cite{me:1,me:2,me:conf,me:ex,me-CJ,me-NATO,me-spher,me:3}; 
for complete theory and its applications to
numerous sciences see monograph \cite{AnKhr}.

Now we recall two results from the 2-adic ergodic theory. The
first one is a 40-year old folklore criteria of measure-preservation/ergodicity in terms of coordinate functions of a T-function.
\subsubsection{Criteria based on algebraic normal forms}
\label{ssec:ANF}
Recall that the \emph{algebraic normal form} (the ANF for
short) of the Boolean function
$\psi_j(\chi_0,\ldots,\chi_j)$ is the representation of this
function via $\oplus$ (addition modulo~$2$, that is, 
``exclusive or'') and $\cdot$ (multiplication modulo~$2$,
that is,  ``and'', or conjunction). In other words,
the ANF of the Boolean function $\psi$ is its representation
in the form
\[
\psi(\chi_0,\ldots,\chi_j)=
\beta\oplus\beta_0\chi_0\oplus\beta_1\chi_1\oplus\ldots
\oplus\beta_{0,1}\chi_0\chi_1\oplus\ldots,
\]
where $\beta,\beta_0,\ldots\in\{0,1\}$. Recall also that the
\emph{weight} of the Boolean function $\psi_j$ in $(j+1)$
variables is the number of $(j+1)$-bit words that
\emph{satisfy} $\psi_j$; that is, the weight is the cardinality
of the truth set of the Boolean function $\psi_j$.

\begin{theorem}[Folklore]
\label{thm:ergBool}
Let a univariate T-function $f$ be represented in the form \eqref{eq:T-uni}. The T-function $f$ is measure-preserving iff
for each $j=0,1,\ldots$ the Boolean function $\psi_j$ in
Boolean variables $\chi_0,\ldots,\chi_j$ is linear with
respect to the variable $\chi_j$; that is, $f$ is measure-preserving
iff the ANF of each $\psi_j$ is of the form
\[
\psi_j(\chi_0,\ldots,\chi_j)
=\chi_j\oplus\phi_j(\chi_0,\ldots,\chi_{j-1}),
\]
where $\phi_j$ is a Boolean function that does not depend on
the variable $\chi_j$.

The univariate T-function $f$ is ergodic  iff, additionally, all Boolean
functions $\phi_j$ are of {odd weight}. {The latter takes
place} iff $\phi_0=1$, and the full degree of the
Boolean function $\phi_j$ for $j\ge 1$ is exactly $j$, that
is, the ANF of $\phi_j$ {contains a monomial
$\chi_0\cdots\chi_{j-1}$}. Thus, $f$ is ergodic iff $\psi_0(\chi_0)=\chi_0\oplus 1$, and
for $j\ge 1$ the ANF of each $\psi_j$ is of the form
\[
\psi_j(\chi_0,\ldots,\chi_j)
=\chi_j\oplus\chi_0\cdots\chi_{j-1}\oplus
\theta_j(\chi_0,\ldots,\chi_{j-1}),
\]
where the weight of $\theta_j$ is even; i.e.,
$\deg\theta_j\le j-1$.
\end{theorem}
For the proof of the folklore Theorem see \cite[Theorem 4.39]{AnKhr}, or\cite[Lemma
4.8]{me:1}, or \cite[Theorem 1]{LinCompT-fun}.

\begin{remark}
Actually the bit-slice technique of Klimov and Shamir introduced in~\cite{KlSh} is just a
re-statement of Theorem~\ref{thm:ergBool}.
\end{remark}

We note that areas of applications of
Theorem~\ref{thm:ergBool} are restricted: Given
a T-function in a form of composition of basic computer instructions,
most often it is infeasible
to find its coordinate functions $\psi_i$. Thus, to
determine with the use of that theorem whether a given
composition of arithmetic and bitwise logical operators is
bijective or transitive is
possible only for relatively simple compositions like the
mapping $x\mapsto x+x^2\OR C$ considered in~\cite{KlSh}. The
latter mapping is \emph{transitive modulo $2^n$ if and only
if $C\equiv 5 \pmod 8$ or $C\equiv 7\pmod 8$};
see~\cite[Example 9.32]{AnKhr}, \cite[Example 3.14]{me:3},  or~\cite[Example
4.9]{me-CJ} for the proof based on
Theorem~\ref{thm:ergBool}.

Earlier in 1999 Kotomina~\cite{kotomina99} applied
Theorem~\ref{thm:ergBool} to prove the following
statement resulting in the so called \emph{add-xor generators},
which are extremely fast: \emph{The T-function
\[
f(x)=(\ldots((((x+c_0)\XOR d_0)+c_1)\XOR d_1)+\cdots
\]
is transitive modulo $2^n$ $(n\ge 2)$ if and only if it is
transitive modulo~$4$.}

The following proposition, whose proof is also based on
Theorem~\ref{thm:ergBool}, gives a method to
construct new invertible T-functions (respectively,
T-functions with a single cycle property), out of given
T-functions:

\begin{proposition}[see \cite{me:3,me-CJ} and \protect{\cite[Proposition 9.29]{AnKhr}}]
\label{prop:compBool}
Let $F$ be an $(n+1)$-variate T-function such that for all
$z_1,\ldots,z_n$ the T-function $F(x,z_1,\ldots,z_n)$ is
measure-preserving with respect to the variable $x$. Then the
composition $$F(f(x),2g_1(x),\ldots,2g_n(x))$$ is measure-preserving
for arbitrary T-functions $g_1,\ldots,g_n$ and any
invertible T-function $f$.

Moreover, if $f$ is ergodic, then
$f(x+4g(x))$, $f(x\XOR (4g(x)))$, $f(x)+4g(x)$, and
$f(x)\XOR (4g(x))$ are ergodic, for
arbitrary T-function $g$.
\end{proposition}

Although Theorem~\ref{thm:ergBool} can be applied to
determine invertibility/single cycle property of some
T-functions, it is highly doubtful that one can prove,
with the use of Theorem~\ref{thm:ergBool} only,
that, e.g., the following function $f$ is a T-function that is ergodic (it is!):
\begin{multline}
\label{eq:wild}
f(x)=2+\frac{x}{3}+\frac{1}{3^x}+2\left(\frac{(x^2+2x)\XOR
(1/3)}{2x+3}\right)^\frac{(x+1)\AND(1/5)}{1-2x}+ \\ 
2\NOT\left(\left(\frac{(x^2-1)\XOR(1/3)}{2x+1}\right)^
\frac{x\AND(1/5)}{5-2x}\right).
\end{multline}

Therefore we need more delicate tools than
Theorem~\ref{thm:ergBool} to study complicated
compositions of basic machine instructions. These tools are provided by $p$-adic analysis and $p$-adic ergodic theory. The first
one
of the said tools are Mahler series.
\subsubsection{Criteria based on Mahler series}
It is not difficult to see that \emph{every} mapping
$f\colon\N_0\to \Z_p$ (or, respectively, $f\colon\N_0\to
\Z$) admits one and only one representation in the form of
so-called \emph{Mahler interpolation series}
\begin{equation}
\label{eq:Mah}
f(x)=\sum^{\infty}_{i=0}a_i{\binom{x}{i}},
\end{equation}
where $\binom{x}{i} =x(x-1)\cdots (x-i+1)/i!$ for
$i=1,2,\ldots$, and $ \binom{x}{0} =1$; $a_i\in \Z_p$
(respectively, $a_i\in \Z$), $i=0,1,2,\ldots $. This
statement can be easily proved directly, substituting
successively $x=0,1,2,\ldots$ to~\eqref{eq:Mah} and
solving the corresponding equation with unknown $a_x$.

Foremost, if $f$ is uniformly continuous on $\N_0$ with
respect to the $p$-adic distance, $f$ can be uniquely
expanded to a uniformly continuous function on $\Z_p$ since
$\Z$ is dense in $\Z_p$. Hence the interpolation series for
$f$ converges uniformly on $\Z_p$. The following is true
(see e.g.~\cite{Mah}): The series $
f(x)=\sum^{\infty }_{i=0}a_i {\binom{x}{i}}$, ($a_i\in
\Z_p$, $i=0,1,2,\ldots$) converges uniformly on $\Z_p$ if
and only if $ \lim^p_{i \to \infty }a_i=0$, where $\lim^p$
is a limit with respect to the $p$-adic distance; hence
uniformly convergent series defines a uniformly continuous
function on $\Z_p$.

The following theorem holds:

\begin{theorem}[see \cite{me:1,me:2} and \protect{\cite[Theorem 3.53]{AnKhr}}]
\label{thm:Mah-comp}
The function $f\colon\Z_p\to \Z_p$ represented
by~\eqref{eq:Mah} is compatible 
if and only if
\[
a_i\equiv 0\pmod {p^{\lfloor\log_pi\rfloor}}
\]
for all $i=2,3,4,\ldots $. \textup{(Here and after for a real
$\alpha $ we denote $\lfloor\alpha \rfloor$ the integral part
of $\alpha $, i.e., the nearest to $\alpha $ rational
integer that is not larger than $\alpha$.)}
\end{theorem}
\begin{remark}
We remind that in the case $p=2$ compatible functions are T-functions,
and vice versa.
\end{remark}
\begin{remark}
\label{rem:log}
Note that the number $\lfloor\log_pi\rfloor$ for
$i=1,2,3,\ldots$ has a very natural meaning: it is the
number of digits in a base-$p$ expansion of $i$,
decreased by $1$. That is, $\lfloor\log_pi\rfloor$ is a bitlength of $i$, decreased by
$1$. So within the context of the paper it is reasonable to
assume that $\lfloor\log_p0\rfloor=0$.
\end{remark}

Now we can give  general characterization of
measure-preserving (resp., ergodic) T-functions:

\begin{theorem}[see \cite{me:1,me:conf} and \protect{\cite[Theorem 4.40]{AnKhr}}]
\label{thm:ergBin}
A map $f\colon\Z_2\to \Z_2$ is a measure preserving
T-function if and only if it can be represented as
\[
f(x)=c_0+x+\sum^{\infty }_{i=1}c_i\,2^{\lfloor \log_2 i
\rfloor +1}\binom{x}{i}. 
\]
The map $f$ is an ergodic T-function if and only if
it can be represented as
\[
f(x)=1+x+\sum^{\infty}_{i=1}c_i2^{\lfloor
\log_2(i+1)\rfloor+1}\binom{x}{i}, 
\]
where $c_0,c_1, c_2 \ldots \in \Z_2$.
\end{theorem}

Using the identity $\Delta\binom{x}{i}=\binom{x}{i-1}$,
where $\Delta$ is a difference operator, $\Delta
u(x)=u(x+1)-u(x)$, we immediately deduce from
Theorems~\ref{thm:Mah-comp}
and~\ref{thm:ergBin} the following \emph{easy method
to construct a measure preserving or ergodic T-function
out of an arbitrary T-function}:

\begin{corollary}[\cite{me:2}, also \protect{\cite[Theorem 4.44]{AnKhr}}]
\label{Delta}
Every ergodic \textup{(}respectively, every measure preserving\textup{)}
T-function $f\colon\Z_2\to \Z_2$ can be represented as
\[
f(x)=1+x+2\cdot\Delta g(x)
\]
\textup{(}respectively, as $f(x)=d+x+2\cdot g(x)$\textup{)} for a suitable
$d\in\Z_2$ and a suitable T-function $g\colon\Z_2\to \Z_2$; and
vice versa, every function $f$ of the above form is, accordingly, ergodic or
measure-preserving T-function.
\end{corollary}
\begin{remark}
Ergodicity of the T-function \eqref{eq:wild} can be proved with the use of
Corollary \ref{Delta}
\end{remark}
\subsubsection{Why  new techniques are needed}
Theorems \ref{thm:ergBool} and \ref{thm:ergBin}
give methods to construct  measure-preserving/ergodic T-functions from
arithmetic and bitwise logical computer instructions; these methods may be
too difficult (or even impossible) to use when a T-function is a composition
of both arithmetic and bitwise logical operations like masking ($\mask(x,c)=x\AND c$; i.e. when the composition includes $i$-th bit value functions $\delta_i(x)$. 
For instance,  not speaking of more complicated
compositions, it is quite difficult to determine measure
preservation/ergodicity even in the simplest case when $f$ is a linear combination
of $\delta_i(x)$: for the corresponding (rather involved) proof see \cite[Theorem
9.20]{AnKhr} or  \cite{me-CJ}. 
 However, functions of this sort are easy to implement in software
 and hardware;
 moreover, they were already used in some 
ciphers, see e.g.  \cite{ABCv3}. There are other techniques  in $p$-adic ergodic
theory than
the already mentioned ones, like the methods that exploit uniform differentiability (we do not mention the ones in the paper; see \cite[Section
4.6]{AnKhr} about these). However, T-functions that include compositions
of bitwise logical operations are rarely uniformly differentiable.

That is one of reasons why we need new techniques to handle
T-functions of this sort. These new techniques are based on representation
of T-functions in the form of van der Put series. 

The other reason is that using the van der Put representations turns out
to be a general way to speed-up T-function-based cryptographic algorithms via time-memory trade-offs since
actually evaluation of a T-function represented by van der Put series uses
just memory calls and integer summations.

It should be stressed however that compared to the known criteria the new ones are not superior in all cases,
for all T-functions: An answer to the question which of the criteria is better to use in order to determine
bijectivity/transitivity of a given T-function  strongly depends on a composition of the T-function. In some cases a particular criterion just works better than others. For
instance, to determine when a linear combination of $\delta_i(x)$ is transitive
one may use either criteria based on Mahler series (Theorem \ref{thm:ergBin}),
or criteria based on ANFs of coordinate functions (Theorem \ref{thm:ergBool})
or new criteria based on van der Put series (Theorem \ref{thm:ergnew}). Then,
exploiting
the criterion based on Mahler series results in a long involved proof
(cf. \cite[Theorem
9.20]{AnKhr}) while application of the ANF-based  criterion seemingly
results in a shorter one, whereas the use of the new criterion implies a very
short proof, cf. Example  \ref{ex:delta}. We believe (and give
some evidence in Section \ref{sec:App}) that the new criteria
are most suitable for T-functions whose compositions involve numerous bitwise
logical
instructions. However,  the ANF-based criteria may be as effective (or even
better) if a T-function is a relatively
short composition of bitwise logical and/or arithmetic instructions, whereas
the criteria based on Mahler series are seemingly more useful for `classical-shaped'
functions (2-adic exponential functions, rational functions, etc.); and for
`smooth' T-functions it is reasonable to try first criteria that exploit
differentiability, cf. \cite[Section 4.6; Subsection 9.2.2]{AnKhr}.
Moreover, results of the current paper, after being announced (without proofs) in \cite{AKY-DAN},
already have stimulated development of ergodic theory in the ring $\F_2[[X]]$ of formal
power series over a two-element field $\F_2$, see recent paper \cite{ShiTao:erg_F_2[[X]]}.
Maybe the criteria developed in the latter paper can also be 
applied to  determine bijectivity/transitivity of certain T-functions as
well since the metric space of all infinite binary sequences with the metric
$d_2$ can also be treated as the metric space $\F_2[[X]]$ (rather than the
metric space $\Z_2$). We believe that no universal `superior' criterion exists;
nonetheless, every new criterion enriches  researchers' toolbox making it more
diverse and therefore giving more flexibility when determining bijectivity/transitivity
of a given T-function.

\subsection{Van der Put series}
\label{sec:vdp}
Now we recall the 
definition and some properties of van der Put series, see e.g. \cite{Mah,Sch}
for  details.
Given a continuous 
function 
$f\: \Z_p\rightarrow \Z_p$, 
there exists 
a unique sequence 
$B_0,B_1,B_2, \ldots $ 
of $p$-adic integers such that 

\begin{equation}
\label{vdp}
f(x)=\sum_{m=0}^{\infty}
B_m \chi(m,x) 
\end{equation}
for all $x \in \Z_p$, where 
\begin{displaymath}
\chi(m,x)=\left\{ \begin{array}{cl}
1, &\text{if}
\ \left|x-m \right|_p \leq p^{-n} \\
0, & \text{otherwise}
\end{array} \right.
\end{displaymath}
and $n=1$ if $m=0$; $n$ is uniquely defined by the inequality 
$p^{n-1}\leq m \leq p^n-1$ otherwise. The right side series in \eqref{vdp} is called the \emph{van der Put series} of the function $f$. Note that
the sequence $B_0, B_1,\ldots,B_m,\ldots$ of \emph{van der Put coefficients} of
the function $f$ tends $p$-adically to $0$ as $m\to\infty$, and the series
converges uniformly on $\Z_p$. Vice versa, if a sequence $B_0, B_1,\ldots,B_m,\ldots$
of $p$-adic integers tends $p$-adically to $0$ as $m\to\infty$, then  the series in the right
part of \eqref{vdp} converges uniformly on $\Z_p$ and thus defines a continuous
function $f\colon \Z_p\to\Z_p$.

The number $n$ in the definition of $\chi(m,x)$ has a very natural meaning;
it is just the number of digits in a base-$p$ expansion of $m\in\N_0$: As
said (see Remark \ref{rem:log}),
\[
\left\lfloor  \log_p m \right\rfloor 
=
\left(\text{the number of digits in a base-} p \;\text{expansion for} \;m\right)-1;
\]
therefore $n=\left\lfloor  \log_p m \right\rfloor+1$ for all $m\in\N_0$ (recall
that we assume  $\left\lfloor  \log_p 0 \right\rfloor=0$). 

 Note that 
the coefficients $B_m$ are
related to the values of the function $f$ in the following way:
Let 
$m=m_0+ \ldots +m_{n-2} p^{n-2}+m_{n-1} p^{n-1}$ be a base-$p$ expansion
for $m$, i.e., 
 $ m_j\in \left\{0,\ldots ,p-1\right\}$, $j=0,1,\ldots,n-1$ and $m_{n-1}\neq 0$, then
\begin{equation}
\label{eq:vdp-coef}
B_m=
\begin{cases}
f(m)-f(m-m_{n-1} p^{n-1}),\ &\text{if}\ m\geq p;\\
f(m),\ &\text{if otherwise}.
 
\end{cases}
\end{equation}
It is worth noticing  also that $\chi (m,x)$ is merely  a characteristic function of the ball $\mathbf B_{p^{-\left\lfloor  \log_p m \right\rfloor-1}}(m)=m+p^{\left\lfloor  \log_p m \right\rfloor-1}\Z_p$
of radius $p^{-\left\lfloor  \log_p m \right\rfloor-1}$ centered at $m\in\N_0$:
\begin{equation}
\label{eq:chi}
\chi(m,x)=\begin{cases}
1,\ &\text{if}\ x\equiv m \pmod{p^{\left\lfloor  \log_p m \right\rfloor+1}};\\
0,\ &\text{if otherwise}
 
\end{cases}
 =
\begin{cases}
1,\ &\text{if}\ x\in\mathbf B_{p^{-\left\lfloor  \log_p m \right\rfloor-1}}(m);\\
0,\ &\text{if otherwise}
 
\end{cases}
\end{equation}

\section{Main results}
In this Section we prove criteria for measure-preservation/ergodicity for
a T-function in terms of van der Put coefficients of the T-function. However,
we start with a van der Put coefficients based criterion for  compatibility
of a continuous $p$-adic map $\Z_p\>\Z_p$. In the case $p=2$ the criterion
yields necessary and sufficient conditions for a map $\Z_2\>\Z_2$ to be a T-function.
\label{sec:main}
\subsection{The compatibility criterion in terms of van der Put coefficients.}
We first  prove the compatibility criterion 
for 
arbitrary map $\Z_p\>\Z_p$ represented by van der Put series.
\begin{theorem}[Compatibility criterion]
\label{thm:comp}
Let a function $f\colon \Z_p\rightarrow \Z_p$ be represented via van der Put series
\eqref{vdp};  then $f$ is compatible \textup{(that is, satisfies the $p$-adic
Lipschitz
condition with a constant 1)} if and only if $|B_m|_p\le
p^{-\left\lfloor \log_p m \right\rfloor}$ for all $m=0,1,2,\ldots$. 
\end{theorem}
In other words, $f$ is compatible if and only if it can be represented as
\begin{equation}
\label{eq:vdp-comp}
f(x)=\sum_{m=0}^{\infty}
p^{\left\lfloor \log_p m \right\rfloor} b_m \chi(m,x),
\end{equation}
for suitable $b_m\in \Z_p$; $m=0,1,2,\ldots$. In particular, \emph{every T-function $f\colon\Z_2\>\Z_2$
can be represented as
\begin{equation}
\label{eq:vdp-t}
f(x)=\sum_{m=0}^{\infty}
2^{\left\lfloor \log_2 m \right\rfloor} b_m \chi(m,x),
\end{equation}
where $b_m\in \Z_2$, $m=0,1,2,\ldots$; and vice versa, the series \eqref{eq:vdp-t}
defines a T-function}.

\begin{proof}[Proof of Theorem \ref{thm:comp}]
To prove the necessity of conditions, take $m\in\N_0$ and consider its base-$p$-expansion
$m = m_0+ \ldots +m_{n-1}p^{n-1}$. Here
$m_j \in \{0, \ldots ,p-1\}$, $m_{n-1} \neq 0$, and  
$n=\lfloor\log_pm\rfloor+1$. As
\[
m_0+ \ldots +m_{n-2}p^{n-2} \equiv m_0+ \ldots +m_{n-2}p^{n-2}+m_{n-1}p^{n-1} \pmod{p^{n-1}};
\]
then
\[
f(m_0+ \ldots +m_{n-2}p^{n-2}) \equiv f(m_0+ \ldots +m_{n-1}p^{n-1})\pmod{p^{n-1}}
\]
by the compatibility of  $f$. From the latter congruence in view of  \eqref{eq:vdp-coef}  it follows
that $B_m \equiv 0 \pmod{p^{n-1}}$ for $m\ge p$; so $|B_m|\le p^{-\lfloor\log_pm\rfloor}$.

Now we prove the sufficiency of conditions. 
As
$\vert B_m \vert _p \leq p^{-\lfloor\log_p m\rfloor}$,
the sequence $B_0,B_1,\ldots$ tends $p$-adically
to 0 as $m\to\infty$ and so the function $f$ is continuous. Hence while proving
that
$|f(x)-f(y)|_p\le|x-y|_p$ for $x,y\in\Z_p$ we may assume that $x,y\in\N_0$
as $\N_0$ is dense in $\Z_p$. Moreover, by same reasons to prove that $f$ satisfies $p$-adic
Lipschitz condition with a constant 1 it suffices only to prove that given
$x\in\N_0$ and $h,n\in\N$,
$|f(x+hp^n)-f(x)|_p\le p^{-n}$.

Let $h=h_0+h_1p+ \ldots +h_\ell p^\ell$ be a base-$p$ expansion of $h\in\N$, and
let
$n_0<n_1<n_2< \ldots< n_k$ be all indices in the latter base-$p$ expansion
such that $h_{n_0},h_{n_1}, \ldots, h_{n_k}$ 
are nonzero; so
$h=h_{n_0}p^{n_0}+h_{n_1}p^{n_1}+ \ldots +h_{n_k}p^{n_k}$.
Now in view of \eqref{eq:vdp-coef} we have that
\begin{multline}
\label{eq:vdp-comp-f}
f(x+hp^n)=
f(x)+[f(x+p^nh_{n_0}p^{n_0})-f(x)]+\\
\sum_{j=1}^k [f(x+p^n(h_{n_0}p^{n_0}+ \cdots +h_{n_j}p^{n_j}))-
f(x+p^n(h_{n_0}p^{n_0}+ \cdots +h_{n_{j-1}}p^{n_{j-1}}))]=\\
f(x)+\sum_{j=0}^k B_{x+h_{n_0}p^{n+n_0}+ \cdots +h_{n_j}p^{n+n_j}}
\end{multline}

However, by our assumption,
$$\vert B_{x+h_{n_0}p^{n+n_0}+ \cdots +h_{n_j}p^{n+n_j}} \vert _p \leq p^{-(n+n_j)};$$
so \eqref{eq:vdp-comp-f} implies that
$\vert f(x+hp^n)-f(x) \vert _p \leq p^{-n}$
due to the strong triangle inequality that holds for the $p$-adic absolute value, cf. \eqref{eq:s-tri}.
\end{proof}

\subsection{The measure-preservation criteria for  T-functions,
in terms of the van der Put coefficients.}
\label{ssec:mes}

Now we prove two criteria of measure-preservation for
T-functions.

\begin {theorem}
\label{thm:vdp-mespres}
Let $f\:\Z_2\rightarrow \Z_2$ be  a T-function  represented via van der Put
series \eqref{vdp}. 
The T-function $f$ is measure-preserving  if and only if the following conditions
hold simultaneously:
\begin{enumerate}
\item
$B_0+B_1\equiv 1 \pmod2$;
\item
$\left| B_m\right|_2=2^{-\lfloor\log_2m\rfloor}$, 
$m=2,3,\ldots$.
\end{enumerate}
\end{theorem}

\begin{proof}
By Corollary \ref{Delta}, the T-function $f$ is measure-preserving if and only if it can be represented in the form 
 $f(x)=d+x+2g(x),$ where $g\:\Z_2\>\Z_2$ 
 is a  T-function and $d\in \Z_2$. 
Now, given
$g(x)=\sum_{m=0}^\infty \tilde{B}_m \chi (m,x)$, 
the van der Put series for the T-function $g$, 
we find van the der Put coefficients of the function 
$f(x)=d+x+2g(x)$.

As the  van der Put series of the T-function $t(x)=x$ is
\begin{equation}
\label{eq:vdp-x}
t(x)=\sum_{m=1}^\infty 2^{\lfloor\log_2m\rfloor} \chi (m,x),
\end{equation} 
and as $d=d\cdot \chi(0,x)+d\cdot \chi (1,x)$, we get:
\begin{multline}
\label{eq:vdp-delta}
f(x)=(d+2\tilde{B_0})\chi (0,x)+(1+d+2\tilde{B_1})\chi(1,x)+\\
\sum_{m=2}^\infty \left(2^{\lfloor\log_2m\rfloor}+2\tilde{B}_m \right)\chi (m,x).
\end{multline}

This proves necessity of conditions of the Theorem since $|\tilde{B}_m|_2\le 2^{-\lfloor\log_2m\rfloor}$
by Theorem \ref{thm:comp} and hence $\left|2^{\lfloor\log_2m\rfloor}+2\tilde{B}_m\right|_2=
2^{-\lfloor\log_2m\rfloor}$ by the strong triangle inequality \eqref{eq:s-tri}.

To prove sufficiency of the conditions,  we note that the condition
$\left| B_m \right| _2=2^{-\lfloor\log_2m\rfloor}$
implies that $B_m=2^{\lfloor\log_2m\rfloor}+2\tilde{B}_m$ for suitable $\tilde B_m\in\Z_2$,
where
\begin{equation}
\label{eq:vdp-g}
\left|\tilde{B}_m \right|_2 \leq 2^{-\lfloor\log_2m\rfloor},
\end{equation}
and the condition 
$B_0+B_1\equiv 1\pmod 2$ implies that
$B_0=d+2\tilde{B_0}$ and 
$B_1=1+d+2\tilde{B_1}$ for suitable $d,\tilde{B_0},\tilde{B_1}\in\Z_2$. 
Now from equations \eqref{eq:vdp-delta} and \eqref{eq:vdp-x} it follows that
$f(x)=d+x+2g(x)$, where $g(x)=\sum_{m=0}^\infty \tilde B_m\chi(m,x)$ is a
T-function by Theorem \ref{thm:comp} in view of inequality \eqref{eq:vdp-g}.
Thus, $f$ is measure-preserving by Corollary \ref{Delta}. 
\end{proof}
From Theorems \ref{thm:comp} and \ref{thm:vdp-mespres} we deduce now the
following
\begin{corollary} 
\label{cor:vdp-mespres}
A map $f\:\Z_2\rightarrow\Z_2$ is a measure-preserving T-function if and only if
it can be represented as
\[
f(x)= b_0\chi(0,x)+ b_1\chi(1,x)+  
\sum_{m=2}^{\infty}
2^{\left\lfloor \log_2 m \right\rfloor} b_m \chi(m,x),
\]
where $b_m\in\Z_2$, and the following conditions hold simultaneously
\begin{enumerate}
\item $b_0+b_1\equiv 1 \pmod2$;
\item $b_m\equiv 1\pmod 2$, $m=2,3,4\ldots$.
\end{enumerate}
\end{corollary}

\subsection{The ergodicity criteria for  T-functions,
in terms of the van der Put coefficients.}
\label{ssec:erg-vdp}
In this subsection, we prove the following criterion of ergodicity for T-functions:
\begin{theorem} 
\label{thm:ergnew}
A T-function $f\: \Z_2\rightarrow \Z_2$ is ergodic if and only if
it can be represented as
\[
f(x)= b_0\chi(0,x)+ b_1\chi(1,x)+  
\sum_{m=2}^{\infty}
2^{\left\lfloor \log_2 m \right\rfloor} b_m \chi(m,x)
\]
for suitable $b_m\in \Z_2$ that satisfy the following conditions:

\begin{enumerate}
\item 
$b_0 \equiv 1 \pmod2$; 
\item $b_0+b_1 \equiv 3 \pmod4$;
\item $|b_m|_2=1$, $m\geq 2$;
\item $b_2+b_3\equiv2\pmod4$;
\item 
$\sum_{m=2^{n-1}}^{2^{n}-1} b_m \equiv 0\pmod4$, $n\geq 3$.
\end{enumerate}
\end{theorem}
To prove the Theorem, we need the following Lemma:
\begin {lemma}[\cite{Yurova}]
\label{le:erg}
Let $f\:\Z_2\rightarrow \Z_2$ be a T-function represented by van der Put
series \eqref{vdp}. Then $f$ is ergodic if and only there exists a sequence
$a_0,a_1,\ldots$  of 2-adic integers such that

\begin{equation}
\label{eq:vdp-f-fin}
B_m= 
\begin{cases}
1+2(a_1-a_0), &\text{if}\ m=0;\\
2(1+a_0+2a_2-a_1), &\text{if}\ m=1;\\
2^{n-1}+2^n a_{m+1}-2^n a_m, &\text{if}\ 2^{n-1}\leq m< 2^n-1,  n\geq 2\\
2^{n-1}+2^{n+1} a_{2^n}- 2^n a_{2^n-1}-2^n a_{2^{n-1}}, &\text{if}\ m=2^n-1, n\geq 2.
\end{cases}
\end{equation}
\end{lemma}

\begin{proof}
By Corollary \ref{Delta}, a T-function $f$ is ergodic if and only if it can be represented as 
$f(x)=1+x+2(g(x+1)-g(x))$,
where $g(x)$ is a suitable T-function.
That is, by Theorem \ref{thm:comp}, 
\begin{equation}
\label{eq:vdp-g1}
g(x)=a_0\chi (0,x)+ \sum_{m=1}^\infty 2^{n_m-1}a_m\chi (m,x)=\sum_{m=0}^\infty \tilde B_m \chi (m,x),
\end{equation}
for suitable $a_0,a_1,a_2\ldots\in\Z_2$, $\tilde B_0,\tilde B_1,\tilde B_2,\ldots\in\Z_2$ (here
$n_m=\lfloor\log_2m\rfloor+1$, $m=1,2,3,\ldots$).

Now, to prove necessity of conditions of the Lemma, we just need to express
the van der Put coefficients of the function $f$ via the coefficients of the
function $g$.
First, we do this for the van der Put coefficients $\bar B_m$ of the T-function 
\begin{equation}
\label{eq:vdp-g(x+1)}
g(x+1)=\sum_{m=0}^\infty \bar B_m\chi (m,x).
\end{equation}
If
$m>1$ then by \eqref{eq:vdp-coef}
$\bar B_m=g(m+1)-g(m+1-q(m))$ where $q(m)=\delta_{n_m-1}(m)2^{n_m-1}$.
If $m\neq 2^{n_m}-1$ then $q(m)=q(m+1)$,
therefore
$\bar B_m=g(m+1)-g(m+1-q(m+1))=\tilde B_{m+1}=2^{n_m-1} a_{m+1}$ by \eqref{eq:vdp-g}
as $n_m=n_{m+1}$ in this case.
If $m=2^{n_m}-1$ then 
$\bar B_m=g(2^{n_m})-g(2^{n_m-1})$ as $q(2^{n_m}-1)=2^{n_m-1}$.
As $\tilde B_{2^n}=g(2^n)-g(0)$ and $\tilde B_{2^{n-1}}=g(2^{n-1})-g(0)$  by \eqref{eq:vdp-g},
we conclude that
$\bar B_{2^n-1}=\tilde B_{2^n}-\tilde B_{2^{n-1}}$.
Finally,  the coefficients $\bar B_0,\bar B_1$ can be found directly from
\eqref{eq:vdp-g(x+1)}: 
$\bar B_0=g(1)=\tilde B_1$, $\bar B_1=g(2)=\tilde B_0+\tilde B_2$.
Now we can find  the van der Put coefficients $\hat B_m$ of the function
$2(g(x+1)-g(x))$; they are:

\begin{equation}
\label{eq:vdp-2delta_g}
\hat B_m= \begin{cases} 
2(\tilde B_1-\tilde B_0), &\text{if}\ m=0;\\
2(\tilde B_0+\tilde B_2-\tilde B_1), &\text{if}\ m=1;\\
2(\tilde B_{m+1}-\tilde B_m),&\text{if}\ 2^{n-1}\leq m< 2^n-1, n\geq 2;\\
2(\tilde B_{m+1}-\tilde B_m-\tilde B_{\frac{m+1}{2}}),&\text{if}\ m=2^n-1, n\geq 2.
\end{cases}
\end{equation}

As $\chi(0,x)+\chi(1,x)=1$ for all $x\in\Z_2$, from \eqref{eq:vdp-t} we derive
the van der Put expansion for the function $x+1$; namely,
\begin{equation}
\label{eq:vdp-(x+1)}
x+1=\chi (0,x)+2\chi (1,x)+\sum_{m=2}^\infty 2^{n_m-1}\chi (m,x). 
\end{equation} 
From \eqref{eq:vdp-g} we have that
$\tilde B_0=a_0$, $\tilde B_m=2^{n-1} a_m$ 
when $2^{n-1}\leq m\leq 2^n-1$, $n=1,2,\ldots$.
Now combining the latter expressions with \eqref{eq:vdp-(x+1)} and \eqref{eq:vdp-g(x+1)}
we  conclude that the van der Put coefficients $B_m$ of the function 
$f(x)=1+x+2(g(x+1)-g(x))$ are of the form \eqref{eq:vdp-f-fin}.

To prove sufficiency of conditions of the Lemma  we just remark that  the
above argument shows that given
expressions \eqref{eq:vdp-f-fin} for the van der Put coefficients of the
function $f$ we can represent
the T-function 
$f$ in the form
$f(x)=1+x+2(g(x+1)-g(x))$ where the  van der Put expansion for  $g$ is
given by \eqref{eq:vdp-g}. That is, the function $g$ is a T-function by Theorem
\ref{thm:comp}; therefore the T-function $f$ is ergodic by Corollary \ref{Delta}.
\end{proof}

Now we are able to prove the following Proposition which actually is 
a criterion of ergodicity for T-functions, in terms of van der Put coefficients:
\begin{proposition}
\label{prop:vdp-erg}
Let 
$f\:\Z_2 \rightarrow \Z_2$ 
be a T-function represented by the van der Put series \eqref{vdp}.
Then $f$ is ergodic if and only if the following conditions are satisfied
simultaneously:

\begin{enumerate}
\item $B_0 \equiv 1 \pmod2$; 
\item $B_0+B_1 \equiv 3 \pmod4$;
\item $\left| B_m \right| _2=2^{-(n-1)}$, $n\geq2$, $2^{n-1}\leq m < 2^n-1$;
\item $\left|\sum_{m=2^{n-1}}^{2^n-1} (B_m-2^{n-1}) \right|_2 \leq 2^{-(n+1)}$, $n\geq 2$.
\end{enumerate}
\end{proposition}
\begin{proof}
By Lemma \ref{le:erg}, if the T-function $f$ is ergodic then its van der Put coefficients $B_m$
can be expressed in the form \eqref{eq:vdp-f-fin}, for suitable $a_0,a_1,a_2,\ldots\in\Z_2$. From \eqref{eq:vdp-f-fin} by direct calculation  we easily prove that conditions
(i)--(iv) of the Proposition are true.

By Lemma \ref{le:erg}, to prove sufficiency of conditions of the Proposition
we must find a sequence of 2-adic integers 
$a_0,a_1,a_2,\ldots$ such that relations \eqref{eq:vdp-f-fin} for the van
der Put coefficients $B_n$ hold. Take arbitrarily $a_0,a_1\in\Z_2$ so that
\begin{equation}
\label{eq:vdp-c}
a_1-a_0=\frac{B_0-1}{2}
\end{equation} 
(cf. the first equation from \eqref{eq:vdp-f-fin}
and condition (i) of the Proposition); then put
\begin{equation}
\label{eq:vdp-ch}
a_2=\frac{B_1+2(a_1-a_0)-2}{4}=\frac{B_1+B_0-3}{4}
\end{equation}
(cf. the second equation from \eqref{eq:vdp-f-fin}). Note that $a_2\in\Z_2$
due to the condition (ii) of the Proposition.
We construct  $a_3,a_4,a_5,\ldots\in\Z_2$ inductively. Denote 
\begin{equation}
\label{eq:vdp-ch0}
\check{B}_m=\frac{B_m-2^{n-1}}{2^n},
\end{equation}
where $n=\lfloor\log_2 m\rfloor+1$, $m\ge 3$;
then $\check{B}_m\in \Z_2$ by condition (iii). Given $a_{2^{n-1}}\in\Z_2$,
for $\alpha=1,2,\ldots, 2^{n-1}-1$
put
\begin{align}
\label{eq:vdp-ch1}
a_{2^{n-1}+\alpha}=&a_{2^{n-1}}+\sum_{m=2^{n-1}}^{2^{n-1}+\alpha-1} \check{B}_m,\\ \label{eq:vdp-ch2}
a_{2^n}=&a_{2^{n-1}}+\frac{1}{2} \sum_{m=2^{n-1}}^{2^n-1} \check{B}_m.
\end{align}
Then 
$a_{2^{n-1}+\alpha}\in\Z_2$ 
by condition (iii) of the Proposition; and  $a_{2^n}\in\Z_2$ by condition
(iv).
Therefore all
$a_0,a_1,a_2,\ldots$ 
are in $\Z_2$.

Now solving system of equations \eqref{eq:vdp-c},\eqref{eq:vdp-ch},\eqref{eq:vdp-ch0}, \eqref{eq:vdp-ch1}, \eqref{eq:vdp-ch2}
with respect to unknowns $B_m$, $m=0,1,2,3,\ldots$, we see that the van der
Put coefficients $B_m$ satisfy conditions \eqref{eq:vdp-f-fin} of Lemma \ref{le:erg}.
Therefore $f$ is ergodic.
\end{proof}
Now we are able to prove Theorem \ref{thm:ergnew}:
\begin{proof}[Proof of Theorem \ref{thm:ergnew}]
Consider the van der Put expansion \eqref{vdp} of the T-function $f$; then
by Theorem \ref{thm:comp}, $B_m=2^{\left\lfloor  \log_2 m \right\rfloor} b_m$, for suitable $b_m\in\Z_2$. It is clear now
that conditions (i), (ii) and (iii) of Proposition \ref{prop:vdp-erg} are equivalent
respectively to 
conditions (i), (ii) and (iii) of Theorem \ref{thm:ergnew}.

Take $2^{n-1}\leq m < 2^n$, $n\ge 2$; thus
$B_m=2^{n-1}b_m$. Then condition (iv) of Proposition \ref{prop:vdp-erg} is
equivalent to the congruence
$\sum_{m=2^{n-1}}^{2^n-1} (B_m-2^{n-1}) \equiv 0\pmod{2^{n+1}}$; which 
is equivalent to the congruence $2^{n-1} \sum_{m=2^{n-1}}^{2^n-1} (b_m-1)\equiv 0\pmod{2^{n+1}}$
as $B_m=2^{n-1}b_m$. However, the latter congruence is equivalent to the
congruence $\sum_{m=2^{n-1}}^{2^n-1} (b_m-1)\equiv 0\pmod4$ which in turn
is equivalent either to the congruence $\sum_{m=2^{n-1}}^{2^n-1} b_m\equiv 0\pmod4$
(when $n\ge 3$) or to the congruence $\sum_{m=2^{n-1}}^{2^n-1} b_m\equiv 2\pmod4$ (when $n=2$). However, the latter two congruences are respectively
conditions
(v) and (iv) of Theorem \ref{thm:ergnew}.                               
\end{proof}

\section{Applications}
\label{sec:App}
In this Section we consider some applications of the above criteria: We give
a new (and short!) proof of ergodicity of a known ergodic T-function which is used
in a filter of ABC stream cipher from \cite{ABCv3},  then we prove ergodicity
of a more complicated T-function (the latter result is new). After that we
explain how to use Theorem \ref{thm:vdp-mespres} (the bijectivity criterion)
in order to construct huge classes of large Latin squares.  Finally we
present a knapsack-like algorithm for fast computation of arbitrary T-function
that use only integer additions and calls to memory. 
\subsection{Examples of ergodic T-functions with masking}
\label{ssec:example}
In this Subsection we consider two example of ergodic T-functions constructed
from additions, multiplications and masking (i.e., the instruction
$\mask(x,c)=x\AND c$). Thus, being implemented as computer programs both these T-functions are fast enough.

The T-function from Example \ref{ex:delta}
is used to construct a filter in ABC stream cipher \cite{ABCv3}; however, the proof
of its ergodicity (which is based on Mahler series and Theorem \ref{thm:ergBin})
is highly technical and complicated, see e.g. \cite[Theorem
9.20]{AnKhr} or \cite{me-CJ}. Although a shorter proof might
be given with the use of ANF-based criteria (Theorem \ref{thm:ergBool}),  below we
give a very short proof by applying the
ergodicity criteria in terms of van der Put series, cf. Theorem \ref{thm:ergnew}.

Let
$x=\chi_0+\chi_1\cdot 2+\ldots +\chi_k\cdot 2^k+\cdots$ 
be a 2-adic representation of  
$x\in \Z_2$; remind that we denote $\delta_k (x)=\chi_k$ (cf. beginning of Subsection \ref{ssec:NA}).
In other words, the value of the function 
$\delta_k\:\Z_2\to\F_2$ at the point $x\in\Z_2$ is the $k$-th binary digit of the base-2 expansion of $x$; so 
$
\sum_{k=0}^\infty 2^k \delta_k(x) = \sum_{k=0}^\infty \chi_k\cdot2^k  = x
$. 
Note that $\delta_k(x)=1$ if and and only if $x$ is congruent modulo $2^{k+1}$
to either of the numbers $2^k,2^k+1,\ldots, 2^{k+1}-1$; so
\begin{equation}
\label{eq:delta-vdp}
\delta_k(x)=\sum_{m=2^k}^{2^{k+1}-1}\chi(m,x),
\end{equation}
as $\chi(m,x)$ is a characteristic function of the ball $\mathbf B_{2^{-\left\lfloor  \log_2 m \right\rfloor-1}}(m)$, see \eqref{eq:chi}.

\begin{example}
\label{ex:delta}
Given a sequence $c,c_0,c_1,c_2,\ldots$ of 2-adic integers, 
the series   
\begin{equation}
\label{eq:delta}
c+\sum_{i=0}^\infty c_i\delta_i(x)
\end{equation}
defines an ergodic  T-function $f\:\Z_2\rightarrow \Z_2$ if and only if the
following conditions hold simultaneously:
\begin{enumerate}
\item $c\equiv 1\pmod{2}$;
\item $c_0\equiv 1\pmod{4}$;
\item $|c_{i}| _{2}=2^{-i}$, for $i=1,2,3,\ldots$.
\end{enumerate}
\end{example}
Indeed, substituting \eqref{eq:delta-vdp} to \eqref{eq:delta} we
obtain the series
$c+\sum_{i=0}^\infty c_i\sum_{m=2^i}^{2^{i+1}-1}\chi(m,x)$;
so the van der Put coefficients are: $B_0=c$, $B_1=c+c_0$, $B_m=c_{\lfloor\log_2m\rfloor}$
for $m\ge 2$. Now from condition (i) of Proposition \ref{prop:vdp-erg} we have that $c\equiv 1\pmod 2$; however, from condition (ii) of Proposition \ref{prop:vdp-erg} we have that $2c+c_0\equiv 3\pmod 4$ which gives us that
$c_0\equiv 1\pmod 4$. Condition (iii) of Proposition \ref{prop:vdp-erg} is
equivalent to the condition $|c_i|_2=2^i$ for $i\ge 1$. Due to these three conditions,
condition (iv) of Proposition \ref{prop:vdp-erg} is satisfied since 
$\sum_{m=2^i}^{2^{i+1}-1}(c_i-2^i)=2^i c_i-2^{2i}\equiv 0\pmod{2^{i+2}}$
for $i\ge 1$. This ends the proof.

Note that under conditions of Example \ref{ex:delta}, the T-function $f(x)=c+\sum_{i=0}^\infty c_i\delta_i(x)$ can be expressed via operations of integer addition,
integer multiplication by constants, and operation of masking,  $\mask(x,2^i)=x\AND
2^i=2^i\delta_i(x)$:
As $c_i=2^id_i$ for suitable $d_i\in\Z_2$, $i=0,1,2,\ldots$,  by conditions (ii)--(iii), we have that
$f(x)=c+\sum_{i=0}^\infty d_i\cdot\mask(x,2^i)$; so the corresponding T-function
$\bar f=f\bmod 2^k$ on $k$ bit words is 
$\bar f(x)=c+\sum_{i=0}^{k-1} d_i\cdot\mask(x,2^i)$.

Note that in the just considered example the coefficients $c,c_0,c_1,\ldots$
do not depend on $x$. Now we consider a more complicated T-function of
this sort where the  coefficients
depend on $x$.
As the least non-negative residue modulo $2^k$ is a special case of $\mask$
instruction, $x\bmod 2^k=\mask(x,2^k-1)$, the T-function
from Example \ref{ex:delta-2} can be expressed via integer additions, multiplications,
and masking.
\begin{example}
\label{ex:delta-2}
The following T-function $f$ is ergodic on $\Z_2$ : 
\[
f(x)=1+\delta_0(x)+6\delta_1(x)+\sum_{k=2}^\infty (1+2(x\bmod{2^k})
)2^k\delta_k(x).
\]
\end{example}
To prove the assertion  we calculate van der Put coefficients $B_m$.
For $m\in\{0,1\}$ we have:
\begin{enumerate}
\item $B_0=f(0)=1$
\item $B_1=f(1)=2$              
\end{enumerate}
Given $m=m_0+2m_1+\ldots +2^{n-2}m_{n-2}+2^{n-1}$, denote $m=\acute{m}+2^{n-1}$,
$n_m=n=\lfloor\log_2m\rfloor+1$.
Then, we calculate the van der Put coefficients for the case $n_m=2$. 
As $m=m_0+2$ in this case, we see that
\[
B_{m_0+2}=f(m_0+2)-f(m_0)=6;
\]
so $B_2=B_3=6$.

Now we proceed with calculations of $B_m$ for the case $n_m=n\geq 3$:

\begin{align*}
B_m&= f(\acute{m}+2^{n-1}) -f(\acute{m})= 1+ \delta_0 (\acute{m}+2^{n-1}) + 6\cdot \delta_1 (\acute{m}+2^{n-1}) +\\
&+ \sum_{k=2}^{n-1} 2^k(1+2((\acute{m}+2^{n-1})\bmod{2^k})) 
\cdot \delta_k(\acute{m}+2^{n-1}) -f(\acute{m})=\\
&= f(\acute{m})+2^{n-1}(1+2\acute{m}) -f(\acute{m}).
\end{align*}

So we conclude that if  
$n_m\geq 3$ then
\[
B_m=2^{n-1}(1+2\acute{m}),
\]
where $\acute{m}= m_0+2m_1+\ldots +2^{n-2}m_{n-2}$.
Finally we get:
\begin{enumerate}
                        \item $B_0=1\equiv 1\pmod 2$; 
                        \item $B_0+B_1=1+2\equiv 3\pmod 4$; 
                        \item $\left| B_m \right| _2=\left|2^{n_m-1}(1+2\acute{m}) \right| _2=2^{-(n_m-1)}$ for $n_m\geq 3,$ and \\
                                                $\left|B_2\right|_2=\left|B_3\right|_2=\left| 6 \right| _2=2^{-1}$,  $n_2=n_3=2$.
                        \item  For $n_m=2,$ we have that        $\left|(B_2-2)+(B_3-2) \right|_2=2^{-(2+1)}$;
and for $n_m =n\geq 3$, we have that
                                                        \begin{multline*}
                                                \left| \sum _{m=2^{n-1}}^{2^n-1} (B_m-2^{n-1})\right|_2 =
                                                \left| \sum _{\acute{m}=0}^{2^{n-1}-1} \left(B_{\acute{m}+2^{n-1}} -2^{n-1} \right) \right|_2 = \\
                                                 \left| \sum _{\acute{m}=0}^{2^{n-1}-1} (2^{n-1}(1+2\acute{m}) -2^{n-1}) \right|_2 = 
                                                \left| 2^n \sum _{\acute{m}=0}^{2^{n-1}-1} \acute{m} \right|_2 = \\
                                                \left| 2^n \cdot \frac{(1+2^{n-1}-1)}{2}\cdot (2^{n-1}-1) \right|_2 \leq 2^{-(n+1)}.
                                                        \end{multline*}
\end{enumerate}
Therefore, under conditions of Example \ref{ex:delta-2} the T-function $f$ is ergodic by Proposition\ref{prop:vdp-erg}.
\subsection{Latin squares}
\label{ssec:Lat}
In this Subsection we explain how one may use the bijectivity criterion (Theorem
\ref{thm:vdp-mespres}) 
to construct Latin squares of order $2^\ell$.
We recall that 
a \emph{Latin square of order $P$} is a $P\times P$ matrix
containing $P$ distinct symbols (usually denoted by $0,1,\ldots,P-1$) such
that each row and column of the matrix contains each symbol 
exactly once. Latins squares are used in numerous applications: For games (recall sudoku), for private communication networks (password distribution), in coding theory, in cryptography (e.g., as stream cipher combiners), etc., see, e.g.,
monographs \cite{Lat_Keed,Lat_discr}.

There is no problem to construct one Latin square (a circulant matrix serves
an obvious example), a problem  is how to write a software that produces a number of large Latin squares; however, this is only a part of the problem. Another part of the problem is that in some constraint environments (e.g., in smart cards)  the whole matrix can not be stored in memory: Given two numbers $a,b\in\{0,1,\ldots, P-1\}$, the software must calculate
the $(a,b)$-th entry of the matrix on-the-fly. 

A number
of
methods have been developed in order to construct Latin squares, see e.g.
the monographs we refer above; however, not all of the methods provide solution
to the said problem  since the methods are based on mappings which
are somewhat slow if implemented in software (as, e.g., are polynomials over large
finite fields). Therefore new methods that are based on `fast' computer instructions
are needed. Methods of the latter sort have been developed by using the 2-adic
ergodic theory, see  \cite[Section 8.4]{AnKhr} where the said theory is applied
to construct Latin squares as well as pairs of orthogonal Latin squares. 
The methods of the mentioned monograph are based on differentiability of
$p$-adic mappings; by using Theorem \ref{thm:vdp-mespres}, methods based on
van der Put series can also be developed. We illustrate the general idea
by a simple example.

A Latin square of order $P$ is just a bivariate
mapping $F\colon(\Z/P\Z)^2\rightarrow\Z/P\Z$ which is bijective
with respect to either variable. Therefore, given a pair of bijective (measure-preserving)
T-functions
$$f(x)=\sum_{m=0}^{\infty}
2^{\left\lfloor \log_2 m \right\rfloor} b_m \chi(m,x)\ \text{and}\
\bar f(y)=\sum_{m=0}^{\infty}
2^{\left\lfloor \log_2 m \right\rfloor} \bar b_m \chi(m,y)
$$ whose van der
Put coefficients satisfy Corollary \ref{cor:vdp-mespres}, the function
$$F(x,y)=
(\sum_{m=0}^{\infty}
2^{\left\lfloor \log_2 m \right\rfloor} (b_m \chi(m,x)+
\bar b_m \chi(m,y)))\bmod p^\ell$$ 
is a Latin square of order $2^\ell$. The idea can be developed further;
however, this implies expanding of the apparatus of van der Put series to
the multivariate case. Development of the corresponding theory can be a subject of a future work but now it is
out of scope of the current paper. We note only that in order to obtain really
fast performance of the corresponding software, methods of fast
evaluation of T-functions are needed. We introduce a method of that kind in the next Subsection.

\subsection{Fast computation of T-functions}
\label{ssec:t-m}
In this Subsection we demonstrate how by using the van der Put representation
of a T-function one could speed-up evaluation of the T-function via time-memory
trade-offs.

Let a T-function $f$ be represented via van der Put series \eqref{vdp}; then
the respective T-function on $k$-bit words is
\begin{equation}
\label{eq:vdp-k-bit} 
\bar f=f\bmod 2^k=\sum_{m=0}^{2^k-1}B_m\chi(m,x) 
\end{equation}
by Theorem \ref{thm:comp},
see \eqref{eq:vdp-t}. Arrange coefficients $B_m$, $m=0,1,\ldots,2^k-1$ into
array $B(f)=[B_m\colon m=0,1,\ldots,2^k-1]$; so the address of the coefficient $B_m$ is $m$, $m=0,1,\ldots,2^k-1$. From  \eqref{eq:vdp-k-bit} and \eqref{eq:chi}
it follows that the value $\bar f(x)$ of the T-function $f$ at $x\in\{0,1,\ldots,2^k-1\}$
is equal to the sum modulo $2^k$ of coefficients $B_m$ for $m=x\bmod{2},
x\bmod{4},\ldots, x\bmod{2^k}$; so
to calculate the output $\bar f(x)$
given input $x\in\{0,1,\ldots,2^k-1\}$ one 
needs $k-1$ additions modulo $2^k$ and $k$ calls to memory.
To calculate $\bar f(x)$ the following procedure may be used
(note that $\mask(x,2^j-1)=x\bmod 2^j$):

{\tt
\begin{tabbing}
        
\=~\=~        \=if \=$\mask(x,1)$\=$\ge1$ \=the\=n $S:=B_1$\\
        
           \>\>\>\>\>\>\>else $S:=B_0$;\\
       \>\>\>$i:=1$;\\
        C:      \>\>\>if $i=k$ then $f(x):=S$ and STOP\\ 
            \>\>\>\>\>\=else 
            $i:=i+1$;\\
                \>\>\>if $\mask(x,2^{i}-1)\ge 2^{i-1}$ \=then $S:=S+ B_{\mask(x,2^{i}-1)}$;\\
                   
                  \>\>\> repeat C.                
\end{tabbing} 
}
It can be easily seen that to compute  $f(x)$ the procedure
uses $k$ memory calls to retrieve relevant coefficients $B_m$,  $k$ compare instructions
$\ge/\not\ge$ of integers,  $k$ maskings $\mask$, 
and $k-1$ integer additions.
Note that if necessary
the compare routine
$\mask(x,2^i-1)\ge2^{i-1}$ may be replaced with the routine to determine whether $\mask(x,2^{i-1})$ is 0 or not; however this doubles the total number of maskings. Note also that given arbitrary measure-preserving (respectively,
ergodic) T-function $f$, the array $B(f)$ consists of $2^k$ integers
which may be too large in practical cases. If so, arrays where most   entries $B_m$  are $2^{\lfloor\log_2m\rfloor}$ may be used then (cf. condition (ii)
of Theorem \ref{thm:vdp-mespres} and condition (iii) of Proposition \ref{prop:vdp-erg},
respectively): In this case,  most entries
must not necessarily be kept in memory, they can be calculated on-the-fly
instead, by a suitable fast
routine.   In connection with the issue
it would be interesting to study other cryptographical properties of measure-preserving/ergodic
T-functions whose van der Put coefficients comprise arrays of this kind.
\section{Conclusion}
\label{sec:Concl}
In the paper, we present new criteria for a T-function to be bijective (that
is, measure-preserving)
or transitive (that is,  ergodic). In the proofs
techniques from non-Archimedean ergodic theory are used: The new criteria
are based on representation of a T-function via van der Put series,  special
series from $p$-adic analysis. 
We note that the criteria are not
`globally'
superior to other known criteria (e.g., the ones based on Mahler series or on ANFs of coordinate
functions): being necessary and sufficient conditions for the bijectivity/transitivity of T-functions, all the criteria (speaking rigorously) are equivalent one to another. However, some criteria are easier to apply to some particular types of T-functions than the other criteria. In the paper we give an evidence
that the criteria based on van der Put
series are most suitable to determine bijectivity/transitivity
of a T-function whose composition includes machine instructions 
like, e.g., masking.
We also use the van der Put series to construct a knapsack-like algorithm
for fast evaluation of T-functions via time-memory trade-offs.


\end{document}